\documentclass[letterpaper,11pt]{article}
\usepackage[utf8]{inputenc}
\usepackage{latexsym}
\usepackage{amssymb}
\usepackage[margin=1in]{geometry}
\usepackage{amsmath,amsfonts,amsthm}
\usepackage{dsfont}
\usepackage{graphicx}
\usepackage{multirow}
\usepackage{ifsym}
\usepackage{color}
\usepackage{dsfont}
\usepackage{xcolor}
\usepackage{algorithm}
\usepackage{algorithmic}
\usepackage{natbib}
\usepackage{hyperref}
\hypersetup{%
   breaklinks,%
   colorlinks=true,%
   linkcolor=black,%
   urlcolor=black,%
   citecolor=[rgb]{0,0,0.45}
}
\usepackage{tikz}

\theoremstyle{plain}

\theoremstyle{plain}
\newtheorem{theorem}{Theorem}[section]
\newtheorem{lemma}[theorem]{Lemma}
\newtheorem{claim}[theorem]{Claim}
\newtheorem{conjecture}[theorem]{Conjecture}
\newtheorem{definition}[theorem]{Definition}
\theoremstyle{definition}
\newtheorem{observation}[theorem]{Observation}
\newtheorem{remark}[theorem]{Remark}

\DeclareSymbolFont{bbold}{U}{bbold}{m}{n}
\DeclareSymbolFontAlphabet{\mathbbold}{bbold}

\newcommand{\reals}{{\mathbb R}}
\newcommand{\eps}{{\varepsilon}}
\newcommand{\1}{\mathds{1}}

\newcommand{\expec}{\mathbb{E}}

\newcommand{\cost}{\mathrm{cost}_p}

\newcommand\gencost{\mathrm{cost}_{p,q}}
\newcommand\vol{{\operatorname{vol}}}

\title{Approximating Fair Clustering with Cascaded Norm Objectives}
\author{Eden Chlamt\'a\v{c}\thanks{Department of Computer Science, Ben-Gurion University. The work was done while the author was visiting and supported in part by TTIC. Email: \href{chlamtac@cs.bgu.ac.il}{chlamtac@cs.bgu.ac.il}} \and Yury Makarychev\thanks{Toyota Technological Institute at Chicago (TTIC). Supported by NSF awards CCF-1718820, CCF-1955173, and CCF-1934843. Email: \href{mailto:yury@ttic.edu}{yury@ttic.edu}} \and Ali Vakilian\thanks{Toyota Technological Institute at Chicago (TTIC). Supported by NSF award CCF-1934843. Email: \href{mailto:vakilian@ttic.edu}{vakilian@ttic.edu}}}
\date{}
\begin{document}

\maketitle

\begin{abstract} 
We introduce the $(p,q)$-Fair Clustering problem. In this problem, we are given a set of points $P$ and a collection of different weight functions $W$. We would like to find a clustering which minimizes the $\ell_q$-norm of the vector over $W$ of the $\ell_p$-norms of the weighted distances of points in $P$ from the centers. This generalizes various clustering problems, including Socially Fair $k$-Median and $k$-Means, and is closely connected to other problems such as Densest $k$-Subgraph and Min $k$-Union. 

We utilize convex programming techniques to approximate the $(p,q)$-Fair Clustering problem for different values of $p$ and $q$. When $p\geq q$, we get an $O(k^{(p-q)/(2pq)})$, which nearly matches a $k^{\Omega((p-q)/(pq))}$ lower bound based on conjectured hardness of Min $k$-Union and other problems. When $q\geq p$, we get an approximation which is independent of the size of the input for bounded $p,q$, and also matches the recent $O((\log n/(\log\log n))^{1/p})$-approximation for $(p, \infty)$-Fair Clustering by Makarychev and Vakilian (COLT 2021).
\end{abstract}
\section{Introduction}
Clustering is one of the fundamental problems in various areas including theoretical computer science, machine learning and operations research. Due to its broad range of applications, different researchers have considered different {\em cluster models}. A typical model in theoretical computer science is the {\em centroid model} which contains the classic $k$-Median, $k$-Means, $k$-Center and more generally, $k$-Clustering with $\ell_p$-norm objective. In $k$-Clustering with $\ell_p$-norm objective, we are given a set of points in a metric space and the goal is to find a set of $k$ centers $C$ so as to minimize the total $\ell_p$ distance of points to $C$. These problems are all known to be NP-hard and admit various efficient approximation algorithms~\citep{gonzalez1985clustering,hochbaum1985best,charikar2002constant,kanungo2004local,gupta2008simpler,li2016approximating,ahmadian2019better}. 

Recently, clustering has been widely studied under fairness constraints to prevent or alleviate the bias and discrimination in the solution constructed by the existing algorithms~\citep{chierichetti2017fair,ahmadian2019clustering,bera2019fair,bercea2019cost,backurs2019scalable,schmidt2019fair,huang2019coresets,kleindessner2019fair,chen2019proportionally,micha2020proportionally,jung2019center,mahabadi2020individual,kleindessner2020notion,brubach2020pairwise}. In particular, the notion we consider in this paper is inspired by the notion of {\em Socially Fair Clustering} introduced by~\citet{abbasi2020fair,ghadiri2020fair}. In Socially Fair Clustering, points in the input belong to different groups and the goal is to pick $k$ centers so as to minimize the maximum clustering cost incurred to any of the different (demographic) groups in the input. 
%In very recent works,~\cite{?} provide tight approximation guarantees of socially fair clustering with any $\ell_p$-norm objective. Further, efficient FPT algorithms have been designed for socially fair $k$-Median and $k$-Means~\cite{?}. 
We note that the objective of Socially Fair Clustering was also previously examined by~\citet{anthony2010plant} in the context of {\em Robust $k$-Clustering} where a set of possible scenarios (i.e., a set of clients) are given in the input and the goal is to find a set of centers which is a good solution for all scenarios.

Here, we consider a generalization of both Socially Fair Clustering and Robust Clustering which is denoted as $(p,q)$-Fair Clustering and is formally defined as below. Besides the fact that $(p,q)$-Fair Clustering generalizes various known problems, the new formulation is a generalization of the classic $k$-Clustering problem and understanding the complexity of the problem under different values of $p$ and $q$ is of interest in itself. 

%understanding the approximability of clustering for different values of parameters $p$ and $q$ is of independent interest. 
\begin{definition}[\bf $(p,q)$-Fair Clustering]\label{def:costs}
An instance $\cal I$ consists of a metric space $([m], d)$ on $m$ points, $n$ different groups of points which are represented by non-negative weight functions $w_1, \ldots, w_n$ (for each $i\in [n]$, $w_i: [m]\rightarrow \mathbb{R}_{\geq 0}$) and a target number of centers $k$. For each $i\in [n]$, the $\ell_p$-cost of group $i$ w.r.t.\ a set of centers $C\subseteq [m]$ is defined as 
\begin{align}
\cost^{\cal I}(C, w_i) = \left(\sum_{j=1}^m w_i(j) \cdot d(j, C)^p\right)^{1/p},
\end{align}
where $d(j, C) = \min_{j'\in C} d(j,j')$. Moreover, %for a set of $k$ centers $C \subseteq [m]$, the %$\gencost(C)$
the $(\ell_p,\ell_q)$-cost of a set of $k$ centers $C\subseteq [m]$ is defined as the $\ell_q$-norm of the vector consisting the $\ell_p$-cost of the groups w.r.t.\ $C$ is defined as
\begin{align}\label{eq:cost-p-q}
\gencost^{\cal I}(C) = \left(\sum_{i=1}^n \cost(C, w_i)^q\right)^{1/q}= \left(\sum_{i=1}^n \left(\sum_{j=1}^m w_i(j) \cdot d(j, C)^p\right)^{q/p}\right)^{1/q}
\end{align}
In $(p,q)$-Fair Clustering, the goal is to find a set of $k$ centers $C\subseteq [m]$ so as to minimize $\gencost^{\cal I}(C)$.
\end{definition} 
\begin{remark}
  In the above definition, we may omit the superscript $\cal I$ indicating the instance relative to which we compute the cost when it is clear from the context.
\end{remark}

By setting $q=\infty$, $(p,q)$-Fair Clustering captures Socially Fair Clustering with $\ell_p$-cost. Moreover, compared to Socially Fair Clustering, the formulation of $(p,q)$-Fair Clustering allows for a ``relaxed'' way of enforcing the {\em fairness} requirement. In particular, by varying the value of $q$ from $p$ to $\infty$, the objective interpolates between the objective of classic $k$-Clustering with $\ell_p$-cost and that of Socially Fair Clustering with $\ell_p$-cost.\footnote{We remark that in Socially Fair Clustering  with $\ell_p$-cost as defined by~\cite{makarychev2021approximation}, the goal is to minimize the maximum over all groups $i$ of $\cost(C, w_i)^p$. However, in $(p,\infty)$-Fair Clustering the goal is to minimize the maximum over all groups $i$ of $\cost(C, w_i)$. Note that an $\alpha$-approximation algorithm for $(p,\infty)$-Fair Clustering implies an $\alpha^p$-approximation guarantee for Socially Fair Clustering with $\ell_p$-cost and vice versa.} Depending on how accurate the group membership information is, the user can set the value of $q$ accordingly and aim for a clustering with a ``reasonable'' fairness constraint w.r.t.\ the extracted group membership information: the more accurate the group membership information gets, the higher value can be assigned to $q$. 

We remark that the special case of $(\infty,1)$-Fair Clustering was previously studied by~\citet{anthony2010plant}\ under the name of {\em Stochastic $k$-Center} where the intuition is that we have a potential set of (client) scenarios and know that one of them is likely to happen but do not know which one.~\citeauthor{anthony2010plant}\ proved that Stochastic $k$-Center is as hard to approximate as the Densest $k$-Subgraph. In particular, $(\infty,1)$-Fair Clustering with 0-1 weight functions can be seen to be equivalent to the Min $s$-Union problem (a generalization of Densest $k$-Subgraph)~\citep{chlamtac2018densest,chlamtavc2017minimizing}, in which we are given a collection of $m$ sets %$S_1,\ldots,S_m\subseteq [n]$
and an integer $s\in[m]$ and the goal is to choose $s$ sets from the input whose union has minimum cardinality. In addition to the hardness result for Stochastic $k$-Center,~\citeauthor{anthony2010plant}\ also provided an $O(\log m)$-approximation for the Stochastic $k$-Median problem which is equivalent to $(1,1)$-Fair Clustering with arbitrary weight functions $w_1, \cdots, w_n$.  

In a different direction,~\cite{goyal2021fpt} designed constant factor approximation FPT algorithms for Socially Fair $k$-Median and $k$-Means and provided hardness results as well.

\begin{remark}\label{rem:approx-lp-clustering}
{\rm
\cite{chakrabarty2019approximation} observed that a slightly modified variant of standard existing algorithms for $k$-Median (e.g.,~\cite{charikar2002constant}) can be applied to obtain a constant factor approximation algorithm for the more general (weighted) $k$-Clustering with $\ell_p$-cost (for $p\in [1,\infty)$). This can also be achieved by applying the algorithm of~\citep{makarychev2021approximation} for $(p,\infty)$-Fair Clustering when the number of groups is one.
}
\end{remark}
\subsection{Our Results and Techniques.} 
In this paper, we design approximation algorithms for $(p,q)$-Fair Clustering for all values of $p,q \in [1, \infty)$. The following theorems are our main contributions in this work. 

\begin{theorem}[$p\geq q$]\label{thm:main-p>q}
In the regime $p\geq q$, there exists a polynomial time algorithm that computes an $O\left(k^{\frac{p-q}{2pq}}\right)$-approximation for $(p,q)$-Fair Clustering. 
\end{theorem}
This approximation may be nearly optimal in the following sense. As we show in Appendix~\ref{sec:hardness}, assuming certain hardness conjectures for Min $s$-Union and related problems, it is in fact not possible to get an algorithm for $(p,q)$-Fair Clustering in the regime $p >q$ with an approximation factor less than $m^{\Omega\left((p-q)/(pq)\right)}$ (note that $k\leq m$). See Theorem~\ref{thm:hardness} for details. Moreover, when $k\ll m$ our algorithm achieves a much better approximation which is independent of $m$. 

\begin{theorem}[$p\le q$]\label{thm:main-p-le-q}
In the regime $p\le q$, there exists a polynomial time algorithm that computes an $O\left(\left(\frac{q}{\ln{(1 + q/p)}}\right)^{1/p}\right)$-approximation for $(p,q)$-Fair Clustering.
\end{theorem}
Note that $(p,p)$-Fair Clustering (i.e., $p=q$) with 0-1 % %properly picked
weight functions %$w_1, \cdots, w_n$ that constitute discrete sets $P_1, \cdots, P_n$ reduces t
is equivalent to the classic $k$-Clustering with $\ell_p$-cost. In this setting, Theorem~\ref{thm:main-p-le-q} gives a constant factor approximation which is essentially optimal.
This bound implies an $O\left((\frac{\log n}{\log \log n})^{1/p}\right)$-approximation for $(p,\infty)$-Fair Clustering\footnote{This is the case, since the $(\ell_p, \ell_\infty)$-cost objective is equal to the $(\ell_p, \ell_{\log n})$-cost objective up to a constant factor.},  which matches the %$\left(e^{O(p)}\cdot\frac{\log n}{\log \log n}\right)$-approximation
recent approximation algorithm of~\citep{makarychev2021approximation} and the hardness result of~\citep{bhattacharya2014new}. %
Thus, for any value of $p$, the approximation guarantee of Theorem~\ref{thm:main-p-le-q} smoothly interpolates between the optimal approximation bounds for the previously studied special cases of $q=\infty$ and $q=p$. 

Finally, we remark that for any $p,q$ it is %always
possible to get an $O(n^{|1/p-1/q|})$ approximation by approximating the outer $\ell_q$-norm in the formulation of $(p,q)$-Fair Clustering with $\ell_p$-norm and solving the obtained instance of $(p,p)$-Fair Clustering with one of the existing constant factor approximation algorithms of $k$-Clustering with $\ell_p$-cost.

\paragraph{Overview of our algorithms.} At a high-level, our approach is to solve a convex programming relaxation of $(p,q)$-Fair Clustering and then round the fractional solution to get an approximate integral solution. 
However, even coming up with an efficient convex program of the problems is non-trivial and introducing such a relaxation is one of our contributions in this paper. 
While a generalization of the standard LP relaxations for clustering problems (e.g., $k$-Median~\citep{charikar2002constant}) results in a natural convex programming relaxation of $(p,q)$-Fair Clustering when $p \leq q$, it is less clear how to even come up with a valid convex program relaxation of $(p,q)$-Fair Clustering for the other scenario, $p \ge q$. In particular, the natural constraint we need to add is not convex when $p > q$. 
In the $p\leq q$ regime, we cannot use the natural convex relaxation, or even the stronger relaxation which generalizes the LP-relaxation for $(p, \infty)$-Fair Clustering used by \citeauthor{makarychev2021approximation}, %Makarychev and Vakilian, %for $(p, \infty)$-Fair Clustering, 
since even the latter has %is
a simple $\Omega(n^{(q-p)/q^2})$ integrality gap construction. To overcome this polynomial lower-bound, we use the \textit{round-or-cut} framework (see~\cite{carr2000,an2017lp,li2017uniform,chakrabarty2019generalized}): we have an exponential family of constraints with a separation oracle based on our rounding algorithm; we add a constraint from the family only if the rounding algorithm detects that it is violated. These constraints are needed to bound different moments of the cost function applied to a set of clusters generated by the first step of our rounding.

Next, we employ a slightly modified version of a reduction technique introduced by~\cite{charikar2002constant} which yields a simplified instance of the same problem along with a corresponding convex programming solution. After performing the reduction on the input instance and solution to the corresponding convex program, we will get a {\em sparsified} instance (i.e., the number of points becomes $O(k)$) along with an adjusted feasible solution that together satisfy several useful properties which relate to our convex programming relaxations and are crucial for our rounding algorithm. %Loosely speaking, besides the sparsity property, the reduction preserves the cost up to a constant factor.
More details on the reduction and the properties guaranteed by it are provided in Section~\ref{sec:reduction}. 

Finally, we perform a %randomized
rounding procedure on the sparsified instance and obtain an approximate integral solution. By the properties of the reduction, using the output integral solution for the sparsified instance, we can construct an integral solution of the original instance without increasing the cost by more than a constant factor. We remark that the analysis of our rounding algorithm crucially relies on the properties guaranteed by the ``non-standard'' constraints we added to the relaxations of the problem in both $p\le q$ and $p \geq q$ regimes. 

\subsection{Paper Organization.}
We start with providing the convex programming relaxations of $(p,q)$-Fair Clustering in Section~\ref{sec:convex-program}. In Section~\ref{sec:reduction}, we concisely state the properties of the reduction by~\citep{charikar2002constant} that is used in our rounding algorithm -- its proof is deferred to Appendix~\ref{apx:proof-Charikar}. Then, in Section~\ref{sec:random}, we describe our rounding algorithm and in Sections~\ref{sec:p-geq-q-approx} and~\ref{sec:q-geq-p-approx} we analyze the approximation guarantee of the rounding algorithm in the regimes $p\ge q$ and $p \le q$ respectively. Finally, in Appendix~\ref{sec:hardness}, we explore the connection of $(p,q)$-Fair Clustering to Min $s$-Union and prove polynomial hardness of approximation for the problem assuming standard hardness conjectures for Min $s$-Union.

\section{Convex Programming Relaxations}\label{sec:convex-program}
We will use somewhat different convex programming relaxations in the two parameter regimes, when $p\geq q$ and $p\leq q$. 
However, the relaxations for both regimes will use the following common assignment/clustering polytope, which is often used for problems such as $k$-Median or Facility Location:
\begin{align}
\mathbf{P_{\mathbf{cluster}}^k}:=\{(x,y)\in&[0,1]^{m\times m}\times[0,1]^m\text{ satisfying the following}\}\nonumber\\
    & \sum_{\ell=1}^m x_{j\ell}= 1 &\forall j\in[m]\label{LP:k-median-weight}\\
    &x_{j\ell}\leq y_\ell:=x_{\ell\ell} &\forall j\in[m],\ell\in[m]\setminus\{j\}\label{LP:center-bound}\\
    &\sum_{j=1}^m y_j \leq k\label{LP:k-median-size}\\
    &x_{j\ell}\geq 0&\forall j,\ell\in [m]\label{LP:k-median-sign}
\end{align}

In the intended solution, $x_{j\ell}$ is a 0-1 variable which is $1$ if and only if $\ell$ is a center and $j$ is assigned to a cluster centered at $\ell$. The variable $y_j=x_{jj}$ is $1$ if and only if $j$ is a center.

\subsection{Convex Relaxation for the Case \texorpdfstring{$p\geq q$}{p>q}.}\label{sec:q-le-p-relaxation}

Our algorithm will use the following convex program. We think of $B$ as a fixed parameter of the convex program, which we can use in a binary search.

\begin{align}
\min\quad &B\nonumber\\
\text{s.t.}\quad &\sum_{i=1}^n z_i\leq B^q\label{LP:q-le-p-obj-fun}\\
& (x,y)\in\mathbf{P^k_{cluster}}\\
&z_i\geq\left(\sum_{j=1}^m w_i(j)\left(\sum_{j'=1}^m d(j,j')^q x_{jj'}\right)^{p/q}\right)^{q/p}&\forall i\in[n]\label{LP:q-le-p-cost}
\end{align}

The variable $z_i$ represents the cost incurred by group $i$, raised to the $q$. The natural constraint to express the connection between $z_i$ and the clustering variables $x_{jj'}$ would be
\begin{equation}
    z_i\geq\left(\sum_{j=1}^mw_i(j)\sum_{j'=1}^m d(j,j')^px_{jj'}\right)^{q/p}.
    \label{LP:cost-example}
\end{equation} However, this is not a convex constraint in the $p>q$ regime, so we further relax this connection and write Constraint~\eqref{LP:q-le-p-cost} instead.

\subsection{Convex Relaxation for the Case \texorpdfstring{$p\leq q$}{p<q}.}\label{sec:p-le-q-relaxation}

To motivate our relaxation, first consider the natural convex relaxation which is identical to the relaxation in the previous section but with Constraint~\eqref{LP:q-le-p-cost} replaced by the more natural Constraint~\eqref{LP:cost-example} above (which is convex in the $p\leq q$ regime).
Consider even a simple case where all distances are 1, and the points can be partitioned into sets $(J_1,J_2)$ such that $y_j=1$ for every $j\in J_2$, and for every $j\in J_1$ there is exactly one $j'\in J_2$ such that $x_{jj'}=\varepsilon$ and $y_j=x_{jj}=1-\varepsilon$, and $x_{jj''}=0$ for all $j''\not\in\{j,j'\}$. Also suppose every weight function $w_i$ is an indicator function for some set $P_i\subseteq J_1$ of cardinality $|P_i|=t$ (see Figure~\ref{fig:gap-example}). Note that in this case, the natural cost constraints would give us $z_i\geq (\sum_{j\in P_i} d(j,j')^p\cdot  \varepsilon)^{q/p} = (\varepsilon\cdot t)^{q/p}$ and $B\geq \Bigl(\sum_{i=1}^n z_i\Bigr)^{1/q} \geq n^{1/q}\cdot (\varepsilon\cdot t)^{1/p}$. It would be natural to apply the following randomized rounding: 
\begin{itemize}
\item assign each $j\in J_1$ to center $j'$ independently with probability $y_{jj'} = \varepsilon$, and otherwise make $j$ a center;
\item make each $j\in J_2$ a center (since $y_j = 1)$.
\end{itemize}

\begin{figure}
\centering
\begin{tikzpicture}
\tikzstyle{vertex}=[circle, draw, fill=red!50,inner sep=0pt, minimum width=1.5mm];

\draw (-3,0) ellipse (1cm and 1.2cm);
\draw (3,0) ellipse (1cm and 1.2cm);
\node[below] at (-1.9,-0.6) {$J_1$};
\node[below] at (4.1,-0.6) {$J_2$};

\draw[<->] (-3,-1.6) to node[above] {\footnotesize $d=1$}(3,-1.6);

\node[vertex](va) at (-3.3, 0.7){};
\node[vertex,label=below:$j$](vj) at (-2.7, 0.1){};
\node[vertex](vb) at (-2.6, -0.8){};
\node[vertex](vc) at (-3.4, -0.4){};

\node[vertex](vap) at (3.3, 0.7){};
\node[vertex,label=below:$j'$](vjp) at (2.7, 0.1){};
\node[vertex](vbp) at (2.6, -0.8){};
\node[vertex](vcp) at (3.4, -0.4){};

\draw[thick,->] (vj)--node[above, align=center] {\footnotesize $x_{jj'} = \varepsilon$}(vjp);

\tikzset{every loop/.style={min distance=5mm,in=45,out=135,looseness=10}}
\draw[->] (vj) edge[loop above,thick] node[above] {\footnotesize $1 - \varepsilon$} (vj);
\end{tikzpicture}
\caption{In this example, the natural convex program has a large integrality gap.}
\label{fig:gap-example}
\end{figure}
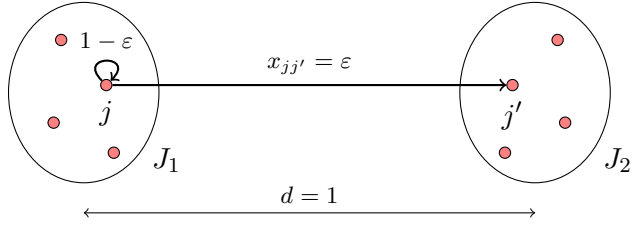

Denote the obtained set of centers by $C$. Let us see what the cost of group $i$ w.r.t. $C$ is. We have $$\expec\Bigl[\cost(C, w_i)^p\Bigr]
= \expec\Bigl[\sum_{j\in P_i} d(j,C)^p \Bigr] = \varepsilon t.$$
Now, we are interested in the expectation $ \expec\Bigl[\left(\sum_{j\in P_i} d(j,C)^p\right)^{q/p} \Bigr]$, as out ultimate goal is to upper bound the expectation of $\gencost(C)^q$ (see formula~(\ref{eq:cost-p-q})).
Observe that if $\varepsilon\cdot t=\Omega(\log n)$, then $\sum_{j\in P_i} d(j,C)^p$ is concentrated around the mean and 
$$ \expec\Bigl[\Bigl(\sum_{j\in P_i}^m d(j,C)^p\Bigr)^{q/p} \Bigr] \sim (\varepsilon t)^{q/p} = O(z_i),$$
as desired. However, if $\varepsilon t = o(1)$, then $\sum_{j=1}^m d(j,C)^p$ is close to a Poisson random variable with rate $\varepsilon t$. Thus, the expectation of $\Bigl(\sum_{j=1}^m d(j,C)^p\Bigr)^{q/p}$ is at least $\sim \varepsilon t$, which is much greater than $(\varepsilon t)^{q/p}$ (in this regime).
Thus, 
$$  
\left(\expec\Bigl[\sum_{i=1}^n\Bigl(\sum_{j=1}^m d(j,C)^p\Bigr)^{q/p}\Bigr]\right)^{1/q} \gtrsim n^{1/q} (\varepsilon t)^{1/q},
$$
which is a factor of $\left(\frac{1}{\varepsilon\cdot t}\right)^{(q-p)/(pq)}$ larger than our bound for $B$.

To overcome this potentially polynomially-large gap, we need to introduce different constraints to handle different moments in the randomized rounding (we will use Latała's inequality in the analysis, which basically allows us to handle the two cases above separately). For instance, to decrease the integrality gap and improve the performance of the rounding w.r.t.\ the relaxation in the above example when $\eps t=o(1)$, we may want to add constraints of the form
$$z_i\geq\sum_{j=1}^m (1-y_j)w_i(j)^{q/p}d(j,[m]\setminus j)^q.$$
However, it is not enough to introduce such constraints for the original points, since we apply our randomized rounding to an instance produced by the reduction of~\citeauthor{charikar2002constant} Loosely speaking, the reduction partitions the set of points $[m]$ into $O(k)$ initial clusters $\{V_\ell\}$; it is guaranteed that there is a $k$-clustering %of cost $O(B^*)$
that for each $\ell$, assigns all points in $V_\ell$ to the same center, and has cost at most a constant times greater than the optimum clustering. % (where $B^*$ is the cost of the optimal clustering).
In our algorithm, for every $\ell$, we will find one center that serves all points in $V_\ell$ using randomized rounding.
For this approach to work, we need to introduce new constraints that depend on the partition $\{V_\ell\}$.
The challenge is that sets $\{V_\ell\}$ returned by the reduction depend on the convex program solution, so we do not know them when we solve the convex program. To deal with this problem, we introduce a family of exponentially-many constraints for all possible choices of sets $\{V_\ell\}$. Since we only need these constraints to hold for the set of clusters arising in the reduction and not all possible collections of clusters, we will check these constraints with a ``rounding separation oracle." That is, we will only check that the constraints hold for the set of clusters arising in the reduction, and if they do not, we will use a separating hyperplane to continue solving the convex program with the Ellipsoid Method. 
In order to describe our non-standard constraints, we introduce the following definition (see Figure~\ref{fig:volume}).

\begin{definition}
Consider a set of points $U\subset [m]$. Let $\vol_i(U) = \sum_{j\in U} w_i(j) d(j, [m] \setminus U)^p$.
\end{definition}
\begin{figure}
\centering
\usetikzlibrary{calc}
\begin{tikzpicture}
\tikzstyle{vertex}=[circle, draw, fill=red!50,inner sep=0pt, minimum width=1.5mm];
\tikzstyle{vertexNU}=[circle, draw, fill=blue!50,inner sep=0pt, minimum width=1.5mm];

\draw (0,0) ellipse (3.5cm and 1.2cm);
\node[below] at (3.3,-0.6) {$U$};

\node[vertex,label=below:$j$](va) at (-0.3, 0.7){};
\node[vertex](vb) at (-2.2, 0.1){};
\node[vertex](vc) at (-2.0, -0.8){};
\node[vertex](vd) at (-1.4, -0.4){};
\node[vertex](ve) at (0.4, -0.6){};
\node[vertex](vf) at (1.25, -0.9){};
\node[vertex](vg) at (1.57, 0.75){};
\node[vertex](vh) at (1.7, -0.3){};
\node[vertex](vi) at (2.9, 0.1){};

\node[vertexNU](vva) at (0.3, 1.8){};
\node[vertexNU] (vvb) at ($([scale=2]vb)$){};
\node[vertexNU] (vvc) at ($([scale=1.5]vc)$){};
\node[vertexNU] (vvd) at ($([xscale=0.7,yscale=4]vd)$){};
\node[vertexNU] (vve) at ($([scale=2.5]ve)$){};
\node[vertexNU] (vvg) at ($([scale=1.8]vg)$){};
\node[vertexNU] (vvh) at ($([scale=2.4]vh)$){};

\draw[<->] (va)--node[xshift=-8mm, yshift=3mm] {\footnotesize $d(j,[m]\setminus U)$}(vva);
\draw[<->,dotted] (vb)--(vvb);
\draw[<->,dotted] (vc)--(vvc);
\draw[<->,dotted] (vd)--(vvd);
\draw[<->,dotted] (ve)--(vve);
\draw[<->,dotted] (vf)--(vve);
\draw[<->,dotted] (vg)--(vvg);
\draw[<->,dotted] (vh)--(vvh);
\draw[<->,dotted] (vi)--(vvh);
\end{tikzpicture}
\caption{To compute $\vol_i(U)$, we find the distance from every point $j$ in $U$ to the closest point outside of $U$.}
\label{fig:volume}
\end{figure}
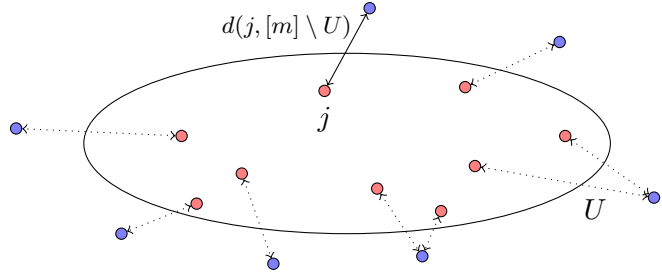

We briefly discuss the motivation for this definition.  Consider a set of points $U \subset [m]$ (later this set will be one of the clusters returned by the reduction). Assume that our algorithm opens a set of centers $C$. By Definition~\ref{def:costs}, $\cost(C, w_i)^p = \sum_{j=1}^m w_i(j) d(j, C)^p$. We want to lower bound the contribution of points in $U$ to $\cost(C, w_i)^p$ assuming that there are no centers from $C$ in $U$. We show that this contribution is at least $\vol_i(U)$. Indeed, observe that $d(j, C) \geq d(j, [m] \setminus U)$ for every $j\in U$, since $C\subseteq [m]\setminus U$. Therefore,
\begin{equation}\label{eq:vol-bound-for-one-set}
\sum_{j\in U} w_i(j) d(j, C)^p  \geq   
\sum_{j\in U} w_i(j) d(j, [m] \setminus U)^p = \vol_i(U).
\end{equation}
In the following claim we generalize this inequality for the case where we have many disjoint sets.

%Our non-standard constraints are based on the following claim:
%\ecnote{The expression $\sum_{j\in V_\ell}w_i(j)d(j,[m]\setminus V_\ell)^{p}$ comes up a lot. Maybe we should define it earlier and give some intuition for how/why we use it.}
\begin{claim}
Consider a collection $(V_\ell)_{\ell\in\Lambda}$ of pairwise disjoint subsets of $[m]$, indexed by some set of indices $\Lambda$. Let $C\subseteq [m]$ be a set of centers. Then for every group $i\in[n]$, the following lower bounds on $\cost(C, w_i)$ hold.
\begin{align}
\cost(C,w_i)^p &\geq \sum_{\ell\in\Lambda: V_\ell \cap C = \emptyset} \vol_i(V_\ell) \quad\text{ and } 
\label{ineq:cluster-bound1}
\\
\cost(C,w_i)^q&\geq \sum_{\ell\in\Lambda: V_\ell \cap C = \emptyset}\vol_i(V_\ell)^{q/p} \text{ if } p \leq q.
\label{ineq:cluster-bound2}
\end{align}

  \label{clm:cluster-bound}
\end{claim}
\begin{proof}
Using Definition~\ref{def:costs}, inequality~(\ref{eq:vol-bound-for-one-set}), and that sets $V_{\ell}$ are pairwise disjoint, we get, 
$$\cost(C, w_i)^p = \sum_{j=1}^m w_i(j) d(j, C)^p \geq \sum_{\ell\in\Lambda: V_\ell \cap C = \emptyset} \sum_{j\in V_\ell} w_i(j) d(j, C)^p \stackrel{\text{\tiny by~(\ref{eq:vol-bound-for-one-set})}}{\geq} 
\sum_{\ell\in\Lambda: V_\ell \cap C = \emptyset} \vol_i(V_\ell),$$
as required.
We obtain inequality~(\ref{ineq:cluster-bound2}) from  inequality~(\ref{ineq:cluster-bound1}) by applying the inequality $\|a\|_1 \geq \|a\|_{q/p}$ (note that $p \leq q$) to the vector $a = (\vol_i(V_\ell))_{\ell \in \Lambda: V_{\ell} \cap C = \emptyset}$.
\end{proof}

We rewrite inequalities~(\ref{ineq:cluster-bound1}) and~(\ref{ineq:cluster-bound2}) as follows. 
\begin{align}
\cost(C,w_i)^q &\geq \left(\sum_{\ell\in\Lambda}\max\left(0,1-|C\cap V_\ell|\right)\vol_i(V_\ell)\right)^{q/p},\label{ineq:cluster-bound1u}\\
\cost(C,w_i)^q &\geq \sum_{\ell\in\Lambda}\max\left(0,1-|C\cap V_\ell|\right)\left(\vol_i(V_\ell)\right)^{q/p}.\label{ineq:cluster-bound2u}
\end{align}
Note that the term $\max\left(0,1-|C\cap V_\ell|\right)$ equals $1$ if $C \cap V_\ell = \emptyset$ and equals $0$ otherwise. To state our convex relaxation, we will use the following notation for disjoint collections of sets:
$$\Pi(m):=\{(\Lambda,(V_{\ell})_{\ell\in\Lambda})\mid \Lambda\subseteq[m]\text{ and } (V_{\ell})_{\ell\in\Lambda} \text{ are disjoint subsets of } [m]\}.$$
We are now ready to present our relaxation for the case $p \leq q$.

\begin{align}
\text{min }& B\nonumber\\
\text{s.t. }&(x,y)\in \mathbf{P^k_{cluster}}\label{LP-q-le-p-first}\\
& \sum_{i=1}^n z_i\leq B^{q} \\
 &z_i\geq\biggl(\sum_{j=1}^mw_i(j)\sum_{j'=1}^m d(j,j')^px_{jj'}\biggr)^{q/p}&\forall i\in[n]
    \label{CP:natural-cost-bound}\\
&z_i\geq \biggl(\sum_{\ell\in\Lambda}\max\bigl(0,1-\sum_{j\in V_\ell}y_j\bigr)\cdot \vol_i(V_\ell)\biggr)^{q/p} &\forall i\in[n],\forall(\Lambda,(V_\ell)_{\ell\in\Lambda})\in\Pi(m)     \label{CP:cluster-bound1}\\
    &z_i\geq \sum_{\ell\in\Lambda}\max\bigl(0,1-\sum_{j\in V_\ell}y_j\bigr)\cdot  \vol_i(V_\ell)^{q/p} &\forall i\in[n],\forall(\Lambda,(V_\ell)_{\ell\in\Lambda})\in\Pi(m)     \label{CP:cluster-bound2}
\end{align}
It follows from
inequalities~(\ref{ineq:cluster-bound1u}) and (\ref{ineq:cluster-bound2u}) that this is a valid relaxation.

As noted in the following remark, it will be easy to check Constraints~\eqref{CP:cluster-bound1} and~\eqref{CP:cluster-bound2} once the set of clusters is fixed.
\begin{remark}

 The correctness of the reduction of \citeauthor{charikar2002constant}\ does not require Constraints~\eqref{CP:cluster-bound1} and~\eqref{CP:cluster-bound2}, so they do not need to be checked before applying the reduction. The reduction will yield a set of centers $K$ along with the corresponding Voronoi cells $(V_\ell)_{\ell\in K}$. To check that Constraints~\eqref{CP:cluster-bound1} and~\eqref{CP:cluster-bound2} hold for these sets for any $\Lambda\subseteq K$, it suffices to check that they hold for $\Lambda=\{\ell\in K\mid \sum_{j\in V_\ell}y_j<1\}$.
\end{remark}
\begin{remark}
Let us briefly compare our relaxation with the one used by~\citeauthor{makarychev2021approximation} for solving $(p,\infty)$-Fair Clustering (let us call the latter the \textit{$(p,\infty)$-relaxation}).
The $(p,\infty)$-relaxation has Constraints~(\ref{LP-q-le-p-first})--(\ref{CP:natural-cost-bound}) and additional non-standard constraints. These non-standard constraints are essentially equivalent to Constraints~(\ref{CP:cluster-bound1}) for a very special family of $(\Lambda, (V_\ell)_{\ell \in \Lambda})$: $|\Lambda| = 1$ and set $V_\ell$ is a ball around some point in $[m]$ (note that Constraints (\ref{CP:cluster-bound1}) and (\ref{CP:cluster-bound2}) are equivalent when $|\Lambda| = 1$). Since there are polynomially many different balls in $([m],d)$, the $(p,\infty)$-relaxation has polynomially-many constraints. 
However, if we simply adapted this relaxation to the $(p,q)$-Fair Clustering problem, we would get a relaxation with a polynomially large integrality gap.
\end{remark}

\section{Reduction à la Charikar et al.}\label{sec:reduction}
We will use a slight modification of the reduction, which was used by~\citeauthor{charikar2002constant}\ to get the first constant factor approximation algorithm for $k$-Median. We use the same reduction to solve the $(p,q)$-Fair Clustering problem in both regimes, $p\leq q$ and $p\geq p$. 

We will refer to a not-necessarily-optimal solution $(x,y,z)$ that satisfies Constraints~(\ref{LP:q-le-p-obj-fun})--(\ref{LP:q-le-p-cost}) when $p\geq q$ and Constraints (\ref{LP-q-le-p-first})--(\ref{CP:natural-cost-bound}) but when $p\leq q$ as a \textit{fractional clustering solution} (when $p \leq q$, the solution might not satisfy Constraints~(\ref{CP:cluster-bound1}) and~(\ref{CP:cluster-bound2}). 

\begin{theorem}
There is a polynomial-time reduction that given an instance $\cal I$ of  $(p,q)$-Fair Clustering, a fractional clustering solution $(x,y,z)$ of value $B$, and a parameter $\gamma \in (0, 1/2)$ returns an instance ${\cal I}'$ on a subset of points $K \subset [m]$ of size $|K| \leq \frac{k}{1-\gamma}$ with weights $w'_i(\ell)$ (where $i\in [n]$ and $\ell \in K$), and a fractional clustering solution $(x',y', z')$ such that the following properties hold.
\begin{enumerate}
\item The cost of $(x',y', z')$ is at most twice that of solution $(x,y,z)$.

\item Let $\sigma(\ell) \in K$ be the closest point to $\ell\in K$ other than $\ell$ itself. The solution $(x',y',z')$ assigns each point $\ell \in K$ only to centers $\ell$ and $\sigma(\ell)$. Specifically, for all $\ell\in K$: $$x'_{\ell\ell} = y'_\ell; \qquad x'_{\ell\sigma(\ell)} = 1 - y'_{\ell}; \qquad
x'_{\ell j} = 0 \text { for }j\notin\{\ell, \sigma(\ell)\}.$$
Further, $x'_{\ell\sigma(\ell)} \leq \gamma$.
\item Let $\{V_\ell\}_{\ell \in K}$ be the Voronoi partition of $[m]$ induced by the set of centers $K$ (Voronoi sites). Then, for every $\ell$,  $y'_\ell=\min\left(1,\sum_{j\in V_\ell}y_j\right)$.%} 

\item The costs of combinatorial solutions for $\cal I$ and ${\cal I}'$ are related as follows: For every set of centers $L\subseteq [m]$ there is a set of centers $L'\subseteq K$ of cost 
\begin{equation}
    \gencost^{{\cal I'}}(L')\leq 2\gencost^{\cal I}(L)+\frac{4}{\gamma^{1/\nu}}B,
    \label{eq:costI-prime-leq-costI}
\end{equation}
and moreover for any set of centers $L\subseteq K$, the cost of $L$ as a solution for $\cal I$ is bounded by
\begin{equation}
    \gencost^{\cal I}(L) \leq \gencost^{{\cal I}'}(L) + \frac{2}{\gamma^{1/\nu}} B.
    \label{eq:costI-leq-costI-prime}
\end{equation}
\item If $\Lambda\subseteq K$ is a set of centers such that $\sigma(\ell)\not\in\Lambda$ for every $\ell\in\Lambda$, then the cost of centers $K\setminus\Lambda$ w.r.t.\ instance $\cal I$ is bounded by
$$\gencost^{\cal I}(K\setminus\Lambda)\leq\frac{6}{\gamma^{1/\nu}}B+2\left(\sum_{i=1}^n\left(\sum_{\ell\in\Lambda} \vol_i(V_\ell)\right)^{q/p}\right)^{1/q}.$$
\end{enumerate}
\label{thm:Charikar-condensed}
\end{theorem}
\citeauthor{charikar2002constant}~presented this reduction and proved that it satisfies properties 1--3 and~\eqref{eq:costI-leq-costI-prime}. As far as we know, Property~\eqref{eq:costI-prime-leq-costI} has not been explicitly stated or used before, but its proof is similar to that of Property~\eqref{eq:costI-leq-costI-prime}. Finally, Property 5 is new -- we need it to use our new CP Constraints~(\ref{CP:cluster-bound1}) and~(\ref{CP:cluster-bound2}).
We will also use the following observation from~\citep{charikar2002constant}.

\begin{observation}
We can efficiently find a partition $(K_1,K_2)$ of $K$ such that $\sigma(K_1)\subseteq K_2$, $\sigma(K_2)\subseteq K_1$, and
  $$\sum_{\ell\in K_1}x'_{\ell\sigma(\ell)}\geq \frac{|K|-k}{2}.$$\label{obs:Charikar}
\end{observation}

For completeness, we provide proofs of Theorem~\ref{thm:Charikar-condensed} and Observation~\ref{obs:Charikar} in Appendix~\ref{apx:proof-Charikar}.

\section{A Randomized Rounding}\label{sec:random}
In this section, we describe two steps that will be used in the rounding algorithms for both parameter regimes. All our algorithms will begin by applying the reduction from Theorem~\ref{thm:Charikar-condensed} for $\gamma=1/5$. Let $K$ be the set of points obtained by applying this reduction, let $(x',y',z')$ be the corresponding convex programming solution, as described in the theorem, and let $K_1$ be as in Observation~\ref{obs:Charikar}. Note that if $|K|\leq k$, then the current centers already give a constant factor approximation (since their cost in the new instance is $0$, and so by Property~\eqref{eq:costI-leq-costI-prime} 
in Theorem~\ref{thm:Charikar-condensed}, the cost in the original instance is at most $2B/\gamma^{1/\nu}\leq 10 B$ (and we can simply add $|K|-k$ centers to $K$ from $[m]\setminus K$ in the original instance only reducing the cost). Thus, we assume that $|K|>k$ below. Consider the following rounding:
\begin{itemize}
    \item Let $$K'=\left\{\ell\in K_1\;\left|\; x'_{\ell\sigma(\ell)}\geq 
    \frac{|K|-k}{4|K_1|}\right.\right\},$$
    and note by Observation~\ref{obs:Charikar} that $$\sum_{\ell\in K'}x'_{\ell\sigma(\ell)}\geq\sum_{\ell\in K_1}x'_{\ell\sigma(\ell)}-\frac{|K|-k}{4|K_1|}\cdot|K_1\setminus K'|\geq
    \frac{|K|-k}{2} - \frac{|K|-k}{4}\geq\frac{|K|-k}{4}.$$
    \item Let $L = K'$. Independently close (remove from $L$) every center $\ell\in K'$ with probability $5x'_{\ell\sigma(\ell)}$ (recall that $x'_{\ell\sigma(\ell)}\leq\gamma=1/5$ by item 2 of Theorem~\ref{thm:Charikar-condensed}).
    \item If $|L|<k$, reopen (add to $L$) an arbitrary collection of $k-|L|$ centers at no additional cost.
\end{itemize}

The expected number of centers that we close in the second step is $5\sum_{\ell\in K'}x'_{\ell\sigma(\ell)}\geq \frac54 (|K|-k)$. By a Chernoff bound this number is at least $|K|-k$ centers w.h.p.\ (unless $|K|-k$ is bounded by some sufficiently large constant, in which case we can enumerate over all possible solutions to get a constant factor approximation). While the analysis of the cost of this rounding depends on the specific parameter regime we are in, we introduce the following notation which we will use in both cases. For every group $i$ and center $\ell\in K'$, we define random variables specifying the per-center and total costs incurred to every group $i\in[n]$ by this rounding:
\begin{equation}
Z_{i\ell}:=\left\{\begin{array}{ll}\vol_i(V_\ell)&\text{if we close center }\ell,\\0&\text{otherwise,}\end{array}\right.\qquad\text{and}\qquad Z_i=\sum_{\ell\in K'}Z_{i\ell}.
\label{def:random-vars-Z-new}
\end{equation}

For the simpler analysis in Section~\ref{sec:p-geq-q-approx}, it will suffice to analyze the cost in the new instance $\cal I'$ (produced by the reduction from Theorem~\ref{thm:Charikar-condensed}), for which we define the following variables:

\begin{equation}
Z'_{i\ell}:=\left\{\begin{array}{ll}w'_i(\ell)d(\ell,\sigma(\ell))^p&\text{if we close center }\ell,\\0&\text{otherwise,}\end{array}\right.\qquad\text{and}\qquad Z'_i=\sum_{\ell\in K'}Z'_{i\ell}.
\label{def:random-vars-Z-prime}
\end{equation}

\iffalse
\begin{equation}
Z_{i\ell}:=\left\{\begin{array}{ll}\sum_{j\in }w'_i(\ell)d(\ell,\sigma(\ell))^p&\text{if we close center }\ell,\\0&\text{otherwise,}\end{array}\right.\qquad\text{and}\qquad Z_i=\sum_{\ell\in K'}Z_{i\ell}.
\label{def:random-vars-Z}
\end{equation}
\fi

\begin{claim}
\label{claim:z-cost}
  The $(\ell_p,\ell_q)$-cost in $\cal I$ of the clustering found by randomized rounding is bounded by\footnote{Note that terms $30B$ and $10B$ only add a constant to the approximation factor of the randomized rounding scheme, since $B$ is at most the cost of the optimal clustering. Thus, the main challenge will be to upper bound either $\left(\sum_{i=1}^nZ_i^{q/p}\right)^{1/q}$ or $\left(\sum_{i=1}^n(Z'_i)^{q/p}\right)^{1/q}$.}
  $$
  \min\left\{
     30 B+2\left(\sum_{i=1}^nZ_i^{q/p}\right)^{1/q},
     10 B+\left(\sum_{i=1}^n(Z'_i)^{q/p}\right)^{1/q}\right\}
  $$
\end{claim}

\begin{proof}
Note that if we close a center $\ell$ then we do not close $\sigma(\ell)$, since (i) we only close centers  $\ell$ in $K'\subseteq K_1$ and (ii) if $\ell\in K'$ then $\sigma(\ell)\in K_2$. 
Therefore, we can apply item 5 of Theorem~\ref{thm:Charikar-condensed}  with $\Lambda = K\setminus L$ (the set of centers we closed). We get,
$$\gencost^{\cal I}(L)\leq\frac{6}{\gamma^{1/\nu}}B+2\left(\sum_{i=1}^n\left(\sum_{\ell\notin L}\vol_i(V_\ell)\right)^{q/p}\right)^{1/q} = \frac{6}{\gamma^{1/\nu}}B+2\left(\sum_{i=1}^n Z_i^{q/p}\right)^{1/q}.
$$
where $\nu = \min(p,q)$, $\gamma = 1/5$, and thus $6/\gamma^{1\nu} \leq 30$.
Also, the cost of $L$ w.r.t. instance ${\cal I}'$ is 
$$\left(\sum_{i=1}^n\left(\sum_{\ell\in K} w_i'(\ell) d(\ell, L)^p\right)^{q/p}\right)^{1/q} = 
\left(\sum_{i=1}^n\left(\sum_{\ell\in K\setminus L} w_i'(\ell) d(\ell, \sigma(\ell))^p\right)^{q/p}\right)^{1/q}
=\left(\sum_{i=1}^n\left(Z_i'\right)^{q/p}\right)^{1/q}.
$$ 
Here, we used that $d(\ell, L)=0$ if $\ell$ is not closed and $d(\ell, L) = d(\ell, \sigma(\ell))$, otherwise.
By Theorem~\ref{thm:Charikar-condensed}, item 4, the cost of $L$ w.r.t.\ instance $\cal I$ is
$$\gencost^{\cal I}(L) \leq 
\frac{2}{\gamma^{1/\nu}}B+\left(\sum_{i=1}^n(Z'_i)^{q/p}\right)^{1/q} \leq 10 B+\left(\sum_{i=1}^n(Z'_i)^{q/p}\right)^{1/q}.$$
\end{proof}

\section{Approximation Algorithm for the Case \texorpdfstring{$p\geq q$}{p>q}}\label{sec:p-geq-q-approx}

In this section, we present our rounding algorithm for the convex program from Section~\ref{sec:q-le-p-relaxation}.

We start out by applying the reduction described in Theorem~\ref{thm:Charikar-condensed}. We now consider two cases. If $|K|-k\geq\sqrt{k}$, then we apply the randomized rounding from Section~\ref{sec:random}. Otherwise, we re-weight the points and run an approximation algorithm for $k$-Clustering with $\ell_q$-norm objective. Let us first analyze the the performance of randomized rounding when $p\geq q$. 

\begin{lemma}
  The expected cost of the randomized rounding is at most $O\left(B\cdot(k/(|K|-k))^{\frac{p-q}{pq}}\right)$.
  \label{lem:analysis-p-geq-q}
\end{lemma}
\begin{proof}
 By Claim~\ref{claim:z-cost}, it suffices to bound the expectation of $\left(\sum_{i=1}^n(Z_i')^{q/p}\right)^{1/q}$. 
Let us do that now. By Jensen's inequality, and linearity of expectation, we have
\begin{align}\expec\left[\left(\sum_{i=1}^n (Z_i')^{q/p}\right)^{1/q}\right]\leq \left(\expec\left[\sum_{i=1}^n (Z_i')^{q/p}\right]\right)^{1/q}&=\left(\sum_{i=1}^n\expec\left[ (Z_i')^{q/p}\right]\right)^{1/q}%\nonumber\\
%&
\leq \left(\sum_{i=1}^n\expec\left[Z_i'\right]^{q/p}\right)^{1/q}\label{eq:q-le-p-rand-round-proof}
\end{align}

Let us now bound the expectation $\expec[Z_i']$. First, recall that by definition of $K'$ and Theorem~\ref{thm:Charikar-condensed}, for every $\ell\in K'$, we have
\begin{equation}
x'_{\ell\sigma(\ell)}\geq\frac{|K|-k}{4|K_1|}\geq\frac{|K|-k}{4|K|}\geq\frac{|K|-k}{4k/(1-\gamma)}=\frac{|K|-k}{5k}.\label{eq:x-prime-pruning}
\end{equation}
From the definition of $Z_i'$, using that $(x',y',z')$ satisfies Constraint~\eqref{LP:q-le-p-cost} by Theorem~\ref{thm:Charikar-condensed}, we have 
\begin{align*}
    \expec[Z_i']&=\sum_{\ell\in K'}\expec\left[Z'_{i\ell}\right]=5\sum_{\ell\in K'}x'_{\ell\sigma(\ell)} w'_i(\ell)d(\ell,\sigma(\ell))^p\\
    &\leq5\left(\frac{5k}{|K|-k}\right)^{\frac{p-q}{q}}\sum_{\ell\in K'}(x'_{\ell\sigma(\ell)})^{p/q} w'_i(\ell)d(\ell,\sigma(\ell))^p &\text{by~\eqref{eq:x-prime-pruning}}\\
    &=5\left(\frac{5k}{|K|-k}\right)^{\frac{p-q}{q}} (z'_i)^{p/q}.&\text{by Constraint~\eqref{LP:q-le-p-cost}}
\end{align*}

Plugging this bound back into~\eqref{eq:q-le-p-rand-round-proof}, we get 
$$\expec\left[\left(\sum_{i=1}^n (Z_i')^{q/p}\right)^{1/q}\right]\leq
\left(5^{q/p} \left(\frac{5k}{|K|-k}\right)^{\frac{p-q}{p}} z_i'\right)^{1/q}=
5^{1/p}\left(\frac{5k}{|K|-k}\right)^{\frac{p-q}{pq}}\cdot (\sum_{i=1}^n z_i')^{1/q},$$
By Constraint~\eqref{LP:q-le-p-obj-fun}
and Theorem~\ref{thm:Charikar-condensed}
$(\sum_{i=1}^n z_i')^{1/q} \leq 2B$. We  get 
$$\expec\left[\left(\sum_{i=1}^n (Z_i')^{q/p}\right)^{1/q}\right]\leq 10\left(\frac{5k}{|K|-k}\right)^{\frac{p-q}{pq}}\cdot B,$$
as required.
\end{proof}

Thus, as mentioned earlier, the randomized rounding indeed gives the desired approximation when $|K|-k\geq \sqrt{k}$. Let us see a different rounding algorithm, which gives the desired guarantee when $|K|-k\leq \sqrt{k}$. In this rounding, we define a new weight function $\hat{w}:K\rightarrow\reals_{\geq 0}$ as follows: $$\hat{w}(\ell):=\sum_{i=1}^n w'_i(\ell)^{q/p}.$$

Recall that for any $q\in [1,\infty)$, $k$-Clustering with $\ell_q$-cost can be approximated up to a constant factor (see Remark~\ref{rem:approx-lp-clustering}). Our rounding algorithm in this case is simple:

\begin{itemize}
    \item Apply a constant-factor approximation for $k$-Clustering with $\ell_q$-cost to the current input (on $K$) with new weights $\hat{w}$ and return the set of centers $L$ chosen by this algorithm. 
\end{itemize}

Let us analyze the approximation guarantee. Recall that the new instance (obtained by the reduction from Theorem~\ref{thm:Charikar-condensed}) has optimum value $B'=O(\gamma^{-1/q})(B^*+B)=O(B^*)$, where $B^*$ is the optimum value of the original instance. Then for the above algorithm applied to the new instance, we have the following guarantee.
\begin{lemma}
  The cost $\gencost^{{\cal I}'}(L)$ is at most $O\left(B'\cdot(|K|-k)^{\frac{p-q}{pq}}\right)$.
  \label{lem:analysis-l-leq-q-k-median}
\end{lemma}
\begin{proof}
Since we use a constant-factor approximation for the $\ell_q$ objective, it suffices to show that for every set $\hat K\subset K$ of $k$ centers, the following values
\begin{itemize}
\item the $(\ell_p, \ell_q)$-cost of $\hat K$ w.r.t.\ the original weights $w_i'$ (from the reduction of Theorem~\ref{thm:Charikar-condensed}), and
\item the $\ell_q$-cost of $\hat K$ w.r.t.\ the new weights $\hat w_i$
\end{itemize}
are within a factor of $(|K|-k)^\frac{p-q}{pq}$ of each other.
Indeed, let $\hat{K}\subseteq K$ be any set of $k$ centers. Then
\begin{align*}
\gencost^{{\cal I}'}(\hat{K})&=\left(\sum_{i=1}^n\left(\sum_{\ell\in K}w'_i(\ell)d(\ell,\hat{K})^p\right)^{q/p}\right)^{1/q}
\\&\leq \left(\sum_{i=1}^n\sum_{\ell\in K}w'_i(\ell)^{q/p}d(\ell,\hat{K})^q\right)^{1/q}&\text{since }\|\cdot\|_{p/q}\leq\|\cdot\|_1\\
&=\biggl(\sum_{\ell\in K}\underbrace{\left(\sum_{i=1}^nw'_i(\ell)^{q/p}\right)}_{\hat w(\ell)}d(\ell,\hat{K})^{q}\biggr)^{1/q},
\end{align*}
which is exactly the $\ell_q$ objective of our new instance applied to this set of centers.

On the other hand, again for any set $\hat{K}\subseteq K$ of $k$ centers,  we can bound the $q$-norm objective of the instance with weights $\hat w$ as follows
\begin{align*}
\biggl(\sum_{\ell\in K}\hat{w}(\ell) d(\ell,\hat{K})^{q}\biggr)^\frac1q &=\left(\sum_{\ell\in K\setminus\hat{K}}\hat{w}(\ell)d(\ell,\hat{K})^{q}\right)^\frac1q =\left(\sum_{\ell\in K\setminus\hat{K}} \sum_{i=1}^nw'_i(\ell)^{q/p}d(\ell,\hat{K})^{q}\right)^\frac1q\\
    &=\left(\sum_{i=1}^n\sum_{\ell\in K\setminus\hat{K}}w'_i(\ell)^{q/p}d(\ell,\hat{K})^q\right)^\frac1q\\
    &\stackrel{\text{\tiny by H\"older}}{\leq}\left(\sum_{i=1}^n\left(\sum_{\ell\in K\setminus\hat{K}} 1^{\frac{p}{p-q}}\right)^{\frac{p-q}{p}}\left(\sum_{\ell\in K\setminus\hat{K}}w'_i(\ell)d(\ell,\hat{K})^{p}\right)^\frac{q}{p}\right)^\frac1q\\
    &=|K\setminus\hat{K}|^{\frac{p-q}{pq}}\left(\sum_{i=1}^n\left(\sum_{\ell\in K}w'_i(\ell)d(\ell,\hat{K})^{p}\right)^{\frac{q}{p}}\right)^\frac1q
    =(|K|-k)^{\frac{p-q}{pq}}\gencost^{{\cal I}'}(\hat{K}),
\end{align*}
and the proof follows.
\end{proof}
\begin{proof}[Proof of Theorem~\ref{thm:main-p>q}.] We solve the convex programming relaxation for the problem and then apply the reduction from Theorem~\ref{thm:Charikar-condensed}. If the set of points $K$ found by the reduction is of cardinality $|K|\leq k$, we add $k-|K|$ points from $[m]$ and return the resulting set, which as noted earlier has cost $O(B)$. Otherwise, if $|K|-k>\sqrt{k}$, we run the randomized rounding procedure, which by Lemma~\ref{lem:analysis-p-geq-q} yields an $O(k^{(p-q)/(2pq)})$-approximation. Finally, if $0<|K|-k\leq\sqrt{k}$, we run a constant-factor approximation for the $\ell_q$ norm $k$-Clustering with weights $\hat w$ and obtain a solution $L$.
Let $B_{{\cal I}'}^*$ be the $(\ell_p, \ell_q)$-cost of the optimal solution for ${\cal I}'$ and $B^*$ be the $(\ell_p, \ell_q)$-cost of the optimal solution for $\cal I$. 
By Lemma~\ref{lem:analysis-l-leq-q-k-median}, 
$$\gencost^{{\cal I}'}(L) \leq 
O(k^{(p-q)/(2pq)}) B_{{\cal I}'}^*.$$
By Theorem~\ref{thm:Charikar-condensed}, item 4:
\begin{align*}\gencost(L) &\stackrel{\phantom{\text{\tiny \eqref{eq:costI-leq-costI-prime}}}}{\leq} \gencost^{{\cal I}'}(L) + 10 B\\
&\stackrel{\text{\tiny \eqref{eq:costI-leq-costI-prime}}}{\leq} O(k^{(p-q)/(2pq)}) B_{{\cal I}'}^* + 10 B\\ &\stackrel{\text{\tiny \eqref{eq:costI-prime-leq-costI}}}{\leq}
O(k^{(p-q)/(2pq)}) (2 B^* + 20B) + 10 B = O(k^{(p-q)/(2pq)}) B^*,
\end{align*}
here we used that $B \leq B^*$. We conclude that the algorithm gives an $O(k^\frac{p-q}{pq})$ approximation, as required.
\end{proof}

\section{Approximation Algorithm for the Case \texorpdfstring{$p\leq q$}{p<q}}\label{sec:q-geq-p-approx}
In this section, we upper bound the cost of the solution produced by the rounding procedure.
In the analysis, we will use Latała's inequality.

\begin{theorem}[\cite{Latala}, Corollary 3]
There exists a universal constant $M$ such that if $Z_1, \dots, Z_N$ are independent non-negative random variables and $\alpha \geq 1$,
then
$$
\left(\expec\Bigl[\Bigl(\sum_{i=1}^N Z_i\Bigr)^\alpha\Bigr]\right)^{1/\alpha}
\leq \frac{M \alpha}{\ln(1 + \alpha)}\max\Bigl(
\sum_{i=1}^N \expec[Z_i] , \Bigl(\sum_{i=1}^N \expec[Z_i^\alpha]\Bigr)^{1/\alpha} \Bigr).
$$
\end{theorem}
\noindent We will use Latała's inequality with $\alpha = q/p$. Denote 
$M_{pq} = \left(\frac{Mq}{p\ln(1 + q/p)}\right)^{q/p}$. Then,
\begin{equation}\label{ineq:latala}
\expec\Bigl[\Bigl(\sum_{i=1}^N Z_i\Bigr)^{q/p}\Bigr]
\leq M_{pq}\Bigl(
\Bigl(\sum_{i=1}^N \expec[Z_i]\Bigr)^{q/p} + \sum_{i=1}^N \expec[Z_i^{q/p}]\Bigr).
\end{equation}
\newcommand{\Kalg}{L}
The algorithm for $(p,q)$-Fair Clustering in the case $p \leq q$ simply solves the convex problem, applies the reduction from Theorem~\ref{thm:Charikar-condensed}, and then runs the randomized rounding procedure. We denote the obtained set of centers by $\Kalg$. 
Now we are ready to prove the main result of this section.
\begin{lemma}The rounding procedure outputs a solution for $\cal I$ of cost at most $O\left(\left(\frac{q}{\ln (1 + q/p)}\right)^{1/p}\right)B$ in expectation.
\end{lemma}
\begin{proof}
We use random variables $Z_{i\ell}$ defined in (\ref{def:random-vars-Z-new}). By Claim~\ref{claim:z-cost}, the cost of the solution for instance ${\cal I}'$ found by the rounding procedure is 
$\left(\sum_{i=1}^n\left(\sum_{\ell\in K'} Z_{i\ell}\right)^{q/p}\right)^{1/q}$.
We upper bound this cost using Latała's inequality~(\ref{ineq:latala}). For every $i\in[m]$, we have
\begin{align*}
\expec\Bigl[\Bigl(\sum_{\ell\in K'}Z_{i\ell}\Bigr)^{q/p}\Bigr] \leq M_{pq} \left(
\Bigl(\expec\bigl[\sum_{\ell\in K'} {Z}_{i\ell }\bigr]\Bigr)^{q/p}+ \sum_{\ell \in K'} \expec[Z_{i\ell }^{q/p}]\right)
\end{align*}
Now,
\begin{align*}
\expec\Bigl[\sum_{i=1}^n\Bigl(\sum_{\ell\in K'}Z_{i\ell}\Bigr)^{q/p}\Bigr] \leq M_{pq} \left(
\sum_{i=1}^n \Bigl(\expec\bigl[\sum_{\ell\in K'} {Z}_{i\ell }\bigr]\Bigr)^{q/p} + \sum_{i=1}^n\sum_{\ell\in K'} \expec[Z_{i\ell}^{q/p}]\right)
\end{align*}
Using that $\expec [Z_{i\ell}] =5x'_{\ell\sigma(\ell)}\cdot \vol(V_\ell)$, we get
 $$\sum_{i=1}^n\Bigl(\expec\bigl[\sum_{\ell\in K'} {Z}_{i\ell}\bigr]\Bigr)^{q/p} \leq
5^{q/p}\sum_{i=1}^n\Bigl(\sum_{\ell\in K'} x'_{\ell \sigma(\ell)} \cdot \vol_i(V_\ell)\Bigr)^{q/p}.$$ By items 2 and 3 of Theorem~\ref{thm:Charikar-condensed}, $x'_{\ell\sigma(\ell)}= 1 - y_{\ell}' =\max(0,1-\sum_{j\in V_\ell}y_j)$, so from Constraint~\eqref{CP:cluster-bound1} we get $$\sum_{i=1}^n\Bigl(\expec\bigl[\sum_{\ell\in K'} {Z}_{i\ell}\bigr]\Bigr)^{q/p} \leq5^{q/p}\sum_{i=1}^n{z_i}\leq 5^{q/p}B^q.$$
We also have, 
$$\sum_{i=1}^n\sum_{\ell\in K'} \expec\bigl[{Z}_{i\ell}^{q/p}\bigr] \leq
5^{q/p}\sum_{i=1}^n\sum_{\ell\in K'}x'_{\ell\sigma(\ell)}\cdot\vol_i(V_\ell)^{q/p}.$$ 
And by the same argument as above, but using Constraint~\eqref{CP:cluster-bound2}, we also get $$\sum_{i=1}^n\sum_{\ell\in K'} \expec\bigl[{Z}_{i\ell}^{q/p}\bigr] \leq5^{q/p}\sum_{i=1}^n {z_i}\leq 5^{q/p}B^q.$$
We conclude that the expected cost of the clustering is at most,
\begin{multline*}
\expec\left[\Bigl(\sum_{i=1}^n\Bigl(\sum_{\ell\in K'} Z_{i\ell}\Bigr)^{q/p}\Bigr)^{1/q}\right] \stackrel{\parbox{10mm}{\tiny\centering Jensen's\\inequality}}{\leq}
\left(\expec\Bigl[\sum_{i=1}^n\left(\sum_{\ell\in K'} Z_{i\ell}\right)^{q/p}\Bigr]\right)^{1/q} \\
\leq
O\Bigl(M_{pq}^{1/q} 
\Bigr) B = O\Bigl(\Bigl(\frac{q}{\ln(1 + q/p)}\Bigr)^{1/p}\Bigr) B.
\end{multline*}
\end{proof}
\begin{proof}[Proof of Theorem~\ref{thm:main-p-le-q}.]
We solve the convex programming relaxation for the problem, apply the reduction, and run the randomized rounding procedure. This gives us a solution of cost at most
$O\left(\left(\frac{q}{\ln (1 + q/p)}\right)^{1/p}\right) B$ in expectation.
\end{proof}

\section*{Acknowledgments}
The authors would like to thank the anonymous reviewers for their helpful suggestions on improving the notation and presentation of these results.

\bibliography{ref}
\bibliographystyle{abbrvnat}

\appendix
\section{Proof of Theorem~\ref{thm:Charikar-condensed} and Observation~\ref{obs:Charikar}}\label{apx:proof-Charikar}
\begin{proof}
(Proof of Theorem~\ref{thm:Charikar-condensed}.)
Let $\nu = \min(p, q)$. We define the Convex Program (CP) cost of point $j\in [m]$ as 
\begin{equation}\label{eq:Cjx}
C(j,x) \equiv C(j) = \left(\sum_{j'\in[m]} x_{jj'} d(j,j')^\nu\right)^{1/\nu}.
\end{equation}
We sort all points according to the value of $C(j)$. Renaming the points if necessary, we may assume that 
$$C(1) \leq \dots \leq C(m).$$
Now we choose a subset of points $K$
and assign each point to exactly one point in $K$. Initially, $K = \varnothing$ and all points are unassigned.
Then we process points one by one, starting with 1 and ending with $m$. When we process point $j$, we perform the following steps if $j$ has not been assigned to any vertex $j'$ yet.
\begin{itemize}
    \item Add $j$ to $K$ and assign $j$ to $j$.
    \item For every unassigned $j' > j$, if $d(j, j') \leq \frac{2}{\gamma^{1/\nu}} C(j')$, assign $j'$ to $j$.
\end{itemize}

After all the points are processed, each of them is assigned to some $\ell\in K$ (some points are assigned to themselves). For $\ell \in K$, let 
\begin{itemize}
    \item $U_\ell$ be the set of points assigned to $\ell$.
    \item $V_\ell$ be the set of points that are closer to $\ell$ than to any other point in $K$ (we break ties arbitrarily).
    \item $\sigma(\ell)$ be the closest point to $\ell$ in $K$ other than $\ell$ itself. We break ties arbitrarily but consistently; then the set of edges $(\ell, \sigma(\ell))$ forms a forest.
\end{itemize}    
For a set of points $A\subseteq [m]$, denote $w_i(A) = \sum_{j\in A} w_i(j)$.
Define new weights $w_\ell'$ for $\ell \in K$ by 
$$w'_i(\ell) = w_i(U_\ell) =\sum_{j\in U_\ell} w_i(j).$$
We obtained the desired instance ${\cal I}'$ of $(p,q)$-Fair Clustering on $K$ with weights $\{w'_i(j)\}_{j\in K}$.
Now we define the CP solution.
\begin{align*}
    y_\ell' &= \min(1, \sum_{j\in V_\ell} y_j)\\
    x_{\ell\ell}' = y_\ell'; \qquad
    x_{\ell\sigma(\ell)}' &= 1 - y_\ell'; \qquad
    x_{\ell j}' = 0&&\text{for } j\notin\{\ell, \sigma(\ell)\}\\
    z_i' &= 2^q z_i
\end{align*}
We verify that all the required properties hold. 
\medskip

\noindent\textbf{We verify that $(x',y',z')$ is a fractional clustering solution and
item 1 holds.} From the definition of $(x',y',z')$,
it follows right away that Constraints (\ref{LP:k-median-weight}),
(\ref{LP:center-bound}), (\ref{LP:k-median-size}), and 
(\ref{LP:k-median-sign}) are satisfied. We need to check that Constraints (\ref{LP:q-le-p-cost}) and (\ref{CP:natural-cost-bound}). We rewrite CP Constraints (\ref{LP:q-le-p-cost}) and (\ref{CP:natural-cost-bound}) uniformly using notation $C(j)$.
\begin{equation}\label{LP:uniform-cost}
    z_i \geq \left(\sum_{j} w_i(j) C(j)^p \right)^{q/p}.
\end{equation}
Note that for all $j\in U_\ell$, $j \geq \ell$ and therefore $C(\ell) \leq C(j)$. We have,
$$
    z_i^{p/q} \geq \sum_{j\in[m]} w_i(j) C(j)^p 
    = \sum_{\ell\in K}\sum_{j\in U_\ell} w_i(j) C(j)^p    
    \geq \sum_{\ell\in K}\sum_{j\in U_\ell} w_i(j) C(\ell)^p=
    \sum_{\ell\in K} w_i'(\ell) C(\ell)^p.
$$
Now consider $\ell\in K$ and $j\notin V_\ell$. Let us say $j\in V_{\ell'}$. Then $d(\ell, j) \geq d(\ell',j)$ and $d(\ell, j) + d(\ell',j) \geq d(\ell, \ell')$. Therefore, 
\begin{equation}
d(\ell, j) \geq d(\ell, \ell')/2 \geq d(\ell, \sigma(\ell))/2
\label{eq:d-ell-j}
\end{equation}
here we used that $\sigma(\ell)$ is the closest to $\ell$ point in $K$ other than $\ell$ itself. We have,
$$C(\ell)^\nu = \sum_{j\in[m]} d(\ell, j)^\nu x_{\ell j}
\geq 
\sum_{j\notin V_\ell} d(\ell, j)^\nu x_{\ell j}
\geq
\sum_{j\notin V_\ell} \left(\frac{d(\ell, \sigma(\ell))}{2}\right)^\nu x_{\ell j} = \left(\frac{d(\ell, \sigma(\ell))}{2}\right)^\nu
\sum_{j\notin V_\ell} x_{\ell j}.
$$
If $y'_\ell <1$ then
\begin{equation*}
\sum_{j\notin V_\ell} x_{\ell j} = 1 - \sum_{j\in V_\ell} x_{\ell j}\geq 1 - \sum_{j\in V_\ell} y_j = 1 - y_{\ell}' = x'_{\ell\sigma(\ell)}.
\end{equation*}
If $y'_\ell = 1$, then $\sum_{j \notin V_\ell} x_{\ell j} \geq 0 = 1 - y_{\ell}' = x'_{\ell\sigma(\ell)}$.
In either case,
\begin{equation}\label{ineq:bound-y}
\sum_{j\notin V_\ell} x_{\ell j} \geq 1 - y_{\ell}' = x'_{\ell\sigma(\ell)}.
\end{equation}
Similarly to~(\ref{eq:Cjx}), define
$C'(\ell, x')\equiv C'(\ell) = \left(\sum_{\ell'\in K} x'_{\ell\ell'} d(\ell,\ell')^\nu\right)^{1/\nu} = d(\ell, \sigma(\ell))\cdot (x'_{\ell\sigma(\ell)})^{1/\nu}$.
Then,
$$C(\ell)^\nu \geq \frac{1}{2^\nu}d(\ell, \sigma(\ell))^{\nu} x'_{\ell\sigma(\ell)} = \frac{1}{2^\nu} C'(\ell)^\nu.$$

We conclude that $C(\ell) \geq C'(\ell)/2$ and hence
$z_i^{p/q} \geq\frac{1}{2^p}\sum_{\ell\in K} w_i'(\ell) C'(\ell)^p$.
Therefore,
$$(z_i')^{p/q} = 2^p \cdot z_i^{p/q} \geq 2^p \cdot \frac{1}{2^p}\sum_{\ell\in K} w_i'(\ell) C'(\ell)^p  = \sum_{\ell\in K} w_i'(\ell) C'(\ell)^p,$$
as required by CP Constraints \eqref{LP:q-le-p-cost} and \eqref{CP:natural-cost-bound}. Note that since $z_i' = 2^q z_i$,
cost $B'$ of solution $(x',y',z')$ is at most twice that of $(x,y,z)$.

\medskip

\noindent\textbf{Now we verify that item 2 holds and $|K| \leq k/(1-\gamma)$.}
Formulas for $x'$ and $y'$ follow from their definitions. First, we show that $1 - y_\ell' \leq \gamma$.
By (\ref{ineq:bound-y}), $1 - y_{\ell}'\leq \sum_{j\notin V_\ell} x_{\ell j}$. Applying (\ref{eq:d-ell-j}), we get
$$1 - y_{\ell}'\leq \sum_{j\notin V_\ell} x_{\ell j} 
\leq \frac{\sum_{j\notin V_\ell} x_{\ell j} d(\ell, j)^{\nu}}{(d(\ell,\sigma(\ell))/2)^\nu} \leq \frac{(2 C(\ell))^{\nu}}{d(\ell,\sigma(\ell))^\nu}.
$$
Now, both points $\ell$ and $\sigma(\ell)$ are in $K$. Therefore, neither $\ell$ was assigned to $\sigma(\ell)$ nor $\sigma(\ell)$ was assigned to $\ell$. This means that $d(\ell, \sigma(\ell)) > \frac{2}{\gamma^{1/\nu}} \max(C(\ell), C(\sigma(\ell))) \geq 
\frac{2C(\ell)}{\gamma^{1/\nu}}$. We conclude that
$1 - y_\ell' \leq \gamma$, as required.
Finally, the bound $|K| \leq k/(1-\gamma)$ follows from the following  inequality $k \geq \sum_{\ell \in K} y'_\ell \geq (1-\gamma) |K|$.

\medskip

\noindent\textbf{It is immediate that item 3 holds. }

\medskip

\noindent\textbf{Now we verify that item 4 holds.} We first prove inequality~(\ref{eq:costI-prime-leq-costI}).
Consider a set of centers $L \subset [m]$. 
%Let $B_{\cal I} = \gencost^{\cal I}(L)$.
Note that points in $L$ are not necessarily in $K$. So $L$ is not necessarily a valid set of centers for instance ${\cal I}'$. However, we can think of $L$ as a set of Steiner centers and then compute the cost of $L$ with respect to instance ${\cal I}'$. Formally, we write
\begin{align*}
\left(\sum_{i=1}^n \left(\sum_{\ell \in K} w_i'(\ell) \cdot d(\ell, L)^p\right)^{q/p}\right)^{1/q} &= 
\left(\sum_{i=1}^n \left(\sum_{\ell \in K}\sum_{j\in U_\ell} w_i(j) \cdot d(\ell, L)^p\right)^{q/p}\right)^{1/q}\\
&\leq 
\left(\sum_{i=1}^n \left(\sum_{\ell \in K}\sum_{j\in U_\ell} w_i(j) \cdot (d(\ell, j) + d(j, L))^p \right)^{q/p}\right)^{1/q}
\end{align*}
Consider function $\|\cdot\|:{\mathbb R}^m \to \mathbb R$ defined by
$$\|v\| = \left(\sum_{i=1}^n \left(\sum_{j=1}^m w_i(j) \cdot |v_j|^p\right)^{q/p}\right)^{1/q}= \left(\sum_{i=1}^n \left(\sum_{\ell \in K}\sum_{j\in U_\ell} w_i(j) \cdot |v_j|^p\right)^{q/p}\right)^{1/q}.$$
Let $\Lambda_i$ be a linear map that sends $v\in {\mathbb R}^m$ to $(w_i(1)^{1/p} v_1, \dots, w_i(j)^{1/p} v_j, \dots, w_i(m)^{1/p} v_m)$.
Note that $\|v\| = \Bigl\| \|\Lambda_1v\|_p , \dots , \|\Lambda_n v\|_p \Bigr\|_q$. Therefore, $\|\cdot\|$ is a seminorm on ${\mathbb R}^m$.
In particular,
\begin{multline*}
\left(\sum_{i=1}^n \left(\sum_{\ell \in K}\sum_{j\in U_\ell} w_i(j) \cdot (d(\ell, j) + d(j, L))^p \right)^{q/p}\right)^{1/q} \\ \leq 
\left(\sum_{i=1}^n \left(\sum_{\ell \in K}\sum_{j\in U_\ell} w_i(j) \cdot d(\ell, j)^p \right)^{q/p}\right)^{1/q}+
\left(\sum_{i=1}^n \left(\sum_{\ell \in K}\sum_{j\in U_\ell} w_i(j) \cdot d(j, L)^p \right)^{q/p}\right)^{1/q}
\end{multline*}
The second term on the right is simply the cost of $L$ with respect to instance $\cal I$ and thus equals $\gencost^{\cal I}(L)$. Now, we upper bound the first term. Since every $j\in U_\ell$ is assigned to $\ell$, $d(\ell, j) \leq \frac{2}{\gamma^{1/\nu}}C(j)$. Thus the first term is upper bounded by
$$\left(\sum_{i=1}^n \left(\sum_{\ell \in K}\sum_{j\in U_\ell} w_i(j) \cdot \Bigl(\frac{2}{\gamma^{1/\nu}} C(j)\Bigr)^p \right)^{q/p}\right)^{1/q}
\stackrel{\text{\tiny by (\ref{LP:uniform-cost})}}{\leq }
\frac{2}{\gamma^{1/\nu}}
\left(\sum_{i=1}^n z_i\right)^{1/q} \leq \frac{2}{\gamma^{1/\nu}} B.$$
We conclude that
\begin{equation*}\label{eq:main-part-six}
\left(\sum_{i=1}^n \left(\sum_{\ell \in K} w_i'(\ell) \cdot d(\ell, L)^p\right)^{q/p}\right)^{1/q} \leq \gencost^{\cal I}(L) + \frac{2}{\gamma^{1/\nu}} B.
\end{equation*}
Now we construct a proper solution $L'$ for $K$. For every point in $j\in L$, we choose a closest to $j$ point $\ell$ in $K$ (breaking ties arbitrarily) and add it to $L'$. It is immediate that for every $\ell\in K$, $d(\ell, L') \leq 2d(\ell, L)$. Thus,
\begin{equation*}
\gencost^{{\cal I}'}(L')=
\Bigr(\sum_{i=1}^n \Bigl(\sum_{\ell \in K} w_i'(\ell) \cdot d(\ell, L')^p\Bigr)^{q/p}\Bigr)^{1/q} \leq 2\Bigl(\gencost^{\cal I}(L) + \frac{2}{\gamma^{1/\nu}} B\Bigr).
\end{equation*}
As required, we constructed a feasible solution for ${\cal I}'$ of cost at most 
$2\gencost^{\cal I}(L) + \frac{4}{\gamma^{1/\nu}} B$.

Now we prove that inequality (\ref{eq:costI-leq-costI-prime}) holds.
The proof is almost identical to that of inequality~(\ref{eq:costI-prime-leq-costI}).
\begin{align*}
\gencost^{\cal I}(L) &= \Big(\sum_{i=1}^n \Big(\sum_{j \in [m]} w_i(j) \cdot d(j, L)^p\Big)^{q/p}\Big)^{1/q} = 
\Big(\sum_{i=1}^n \Big(\sum_{\ell \in K}\sum_{j\in U_\ell} w_i(j) \cdot d(j, L)^p\Big)^{q/p}\Big)^{1/q}\\
&\leq 
\Big(\sum_{i=1}^n \Big(\sum_{\ell \in K}\sum_{j\in U_\ell} w_i(j) \cdot (d(\ell, j) + d(\ell, L))^p \Big)^{q/p}\Big)^{1/q}\\
&\leq 
\Big(\sum_{i=1}^n \Big(\sum_{\ell \in K}\sum_{j\in U_\ell} w_i(j) \cdot d(\ell, j)^p \Big)^{q/p}\Big)^{1/q}
+
\Big(\sum_{i=1}^n \Big(\sum_{\ell \in K}\sum_{j\in U_\ell} w_i(j) \cdot d(\ell, L)^p \Big)^{q/p}\Big)^{1/q}
\end{align*}
As we showed above, the first term is at most $\frac{2}{\gamma^{1/\nu}} B$; the second term equals $\gencost^{{\cal I}'}(L)$. Therefore, 
$$\gencost^{{\cal I}}(L) \leq \gencost^{{\cal I}'}(L) + \frac{2}{\gamma^{1/\nu}} B.$$

\medskip

\noindent\textbf{Now we verify that item 5 holds.} Take a point $\ell\in\Lambda$ and point $j\in V_\ell$. Let $j'$ be the closest point to $j$ outside $V_\ell$ and let $\ell'$ be the closest point in $K$ to it (so $j'\in V_{\ell'})$. Then we have
\begin{align*}
    d(j,K\setminus\Lambda)&\leq d(j,\sigma(\ell))&\text{since }\sigma(\ell)\not\in\Lambda\\
    &\leq d(j,\ell)+d(\ell,\sigma(\ell)) &\text{by triangle inequality}\\
    &\leq d(j,\ell)+d(\ell,\ell') &\text{by choice of }\sigma(\ell)\\
    &\leq d(j,\ell)+2d(j',\ell) &\text{by~\eqref{eq:d-ell-j}}\\
    &\leq d(j,\ell)+ 2d(j',j)+2(j,\ell)&\text{by triangle inequality}\\
    &= 3d(j,\ell)+2d(j,j')= 3d(j,K)+2d(j,[m]\setminus V_\ell)
\end{align*}

\iffalse
By the second inequality above, we also have
$$d(j,K\setminus\Lambda)\leq d(j,K)+d(\ell,\sigma(\ell))\leq 3d(j,K)+2d(\ell,\sigma(\ell)).$$
\fi

Of course, if $\ell\in K\setminus\Lambda$ and $j\in V_\ell$, then $d(j,K\setminus\Lambda)=d(j,K)$. Thus, for \emph{any} center $\ell\in K$ we have
$$d(j,K\setminus\Lambda)\leq 3d(j,K)+2\cdot\1_{\ell\in\Lambda}\cdot d(j,[m]\setminus V_\ell).$$
Using that $\|\cdot\|$ is a seminorm as in the proof of item 4, we get
\begin{align*}
\gencost^{\cal I}&(K\setminus \Lambda) = \Bigl(\sum_{i=1}^n \Bigl(\sum_{j=1}^m w_i(j) \cdot d(j, K\setminus \Lambda)^p\Bigr)^{q/p}\Bigr)^{1/q}\\
&\leq 
3\Bigl(\sum_{i=1}^n \Bigl(\sum_{j=1}^m w_i(j) \cdot d(j, K)^p\Bigr)^{q/p}\Bigr)^{1/q}+
2\Bigl(\sum_{i=1}^n \Bigl(\sum_{\ell\in\Lambda}\sum_{j\in V_\ell} w_i(j) \cdot d(j, [m]\setminus V_\ell)^p\Bigr)^{q/p}\Bigr)^{1/q}\\
&=
3 \gencost^{\cal I}(K) + 2\Bigl(\sum_{i=1}^n \Bigl(\sum_{\ell\in\Lambda}\vol_i(V_\ell)\Bigr)^{q/p}\Bigr)^{1/q}
\end{align*}
By item 4, the cost %$\gencost(K)$
of $K$ w.r.t. instance $\cal I$, that is, $\gencost^{\cal I}(K)$, is at most the cost of $K$ w.r.t. instance ${\cal I}'$, which is 0, plus $2\gamma^{-1/\nu} B$. We conclude that
$$
\gencost(K\setminus \Lambda) \leq 6\gamma^{-1/\nu} B + 2\Bigl(\sum_{i=1}^n \Bigl(\sum_{\ell\in\Lambda}\vol_i(V_{\ell})\Bigr)^{q/p}\Bigr)^{1/q}.$$
\end{proof}
\begin{proof}(Proof of Observation~\ref{obs:Charikar}.)
Recall that $\sigma(\ell)$ is the closest to $\ell$ point in $K$ other than $\ell$ itself. We break ties arbitrarily but consistently so that the set of edges $(\ell, \sigma(\ell))$ forms a forest on $K$. We choose an arbitrary root in every tree in the forest $(K, {(\ell, \sigma(\ell))}$. Then we let $K_1$ be the set of vertices of odd depth and $K_2$ be the set of vertices of even depth.
Clearly, $\sigma(K_1) \subseteq K_2$ and $\sigma(K_2) \subseteq K_1$. 
Now,
$$\sum_{\ell\in K} x'_{\ell\sigma(\ell)} = \sum_{\ell\in K} (1 - y_{\ell}) = |K| - \sum_{\ell\in K} y_{\ell} = |K| - k.
$$
Therefore, $\sum_{\ell\in K_1} x'_{\ell\sigma(\ell)} \geq \frac{|K|-k}{2}$ or 
$\sum_{\ell\in K_2} x'_{\ell\sigma(\ell)} \geq \frac{|K|-k}{2}$.
We are done if the former inequality holds. Otherwise, we simply swap $K_1$ and $K_2$.
\end{proof}

\section{Connection to Min \texorpdfstring{$s$}{s}-Union and Conjectured Hardness}\label{sec:hardness}
Let us see how known hardness conjectures for other problems imply polynomial hardness of approximation for $(p,q)$-Fair Clustering when $q<p-\Omega(1)$.

In the Min $s$-Union problem, we are gives a collection of $m$ sets $S_1,\ldots,S_m\subseteq[n]$, and a parameter $s\in[m]$, and wish to find a collection $J\subseteq[m]$ of (indices of) sets of size $|J|=s$ as to minimize the cardinality of their union $|\bigcup_{j\in J}S_j|$. It has been conjectured (see~\cite{chlamtavc2017minimizing}) that the following distinguishing problem, slightly rephrased here, is hard (and implies polynomial hardness of approximation for Min $s$-Union):

\begin{conjecture}[Dense versus Random Conjecture]
For any constants $0<\varepsilon<\delta<1$, it is hard to distinguish between the following cases when $s\leq\sqrt{m}$:
\begin{itemize}
    \item \textbf{Random:} $n=m^\varepsilon$ and each of the sets $S_1,\ldots,S_m$ is sampled independently by selecting $10\ln m$ elements in $[n]$ independently at random.
    \item \textbf{Dense:} $n=10 m^{\varepsilon}\ln m$, and an adversary may select the sets, with an additional guarantee that there exist sets $J\subseteq[m]$ and $I\subseteq[n]$ such that $|J|=s$ and $|I|=(10s\ln m)^\delta$, and for every $j\in J$, $S_j\subseteq I$ (that is, there are $s$ input sets contained in a single set of cardinality~$\tilde{O}(s^\delta)$).
\end{itemize}
\label{conj:dense-vs-random}
\end{conjecture}

This conjecture is supported by Sherali-Adams integrality gaps (cf.\ \cite{chlamtavc2017minimizing}, \cite{chlamtac-manurangsi18}) and is related to similar conjectures for related problems, and in particular to the Projection Games Conjecture. \cite{DkS-SoS-gaps} show a somewhat weaker (but still polynomial) lower bound for this problem is also supported by a matching integrality gap in the Sum of Squares hierarchy.

A straightforward reduction gives the following:
\begin{theorem}
  For any $1\leq q\leq p$ and constants $0<\varepsilon<\delta<1$, where $\varepsilon\leq\frac12$, if the Dense versus Random Conjecture holds for $\delta,\varepsilon$, then $(p,q)$-Fair Clustering is hard to approximate to within less than $m^{\varepsilon(1-\delta)(p-q)/(pq)}$.\label{thm:hardness}
\end{theorem}
\begin{proof}
Let us see a reduction from Min $s$-Union with these parameters to $(p,q)$-Fair Clustering, and analyze the optimum clustering value in both cases.

Given an instance $(S_1,\ldots,S_m,s)$ of Min $s$-Union, construct a clustering instance as follows: Identify the points $[m]$ with the sets $S_1,\ldots, S_m$ and let the distance between any two points be $1$. For all $i\in[n]$ and $j\in[m]$, let $w_i(j)=\1_{i\in S_j}$. Finally, let the target number of centers be $k=m-s$. In the context of the reduction, let us consider instances for which $s=m^{\varepsilon}(\leq\sqrt{m})$.

Consider first the random case. By a Chernoff bound, w.h.p., for any set $J\subseteq m$ of cardinality $|J|=s=m^{\varepsilon}$, we have
$$\left|\bigcup_{j\in J}S_j\right|\geq \frac12m^{\varepsilon}\ln n.$$
Thus, for any set of centers $K$ of cardinality $|K|=k$, the cost of $K$ is bounded from below by
\begin{align*}\gencost(K)&=\left(\sum_{i=1}^n\left(\sum_{j=1}^m\1_{i\in S_j}d(j,K)^p\right)^{q/p}\right)^{1/q}\\
&=\left(\sum_{i=1}^n|\{j\in [m]\setminus K\mid i\in S_j\}|^{q/p}\right)^{1/q}\\
&\geq \left(\sum_{i=1}^n\1_{i\in\bigcup_{j\in[m]\setminus K}S_j}\right)^{1/q}\\
&=\left|\bigcup_{j\in [m]\setminus K}S_j\right|^{1/q}\geq \left(\frac12m^{\varepsilon}\ln m\right)^{1/q}.
\end{align*}

On the other hand, consider the dense case. Let $J$ be the ``dense" set as defined in the conjecture, and let $K=[m]\setminus J$. Then as before, the cost of these centers is bounded from above by
\begin{align*}\gencost(K)&=\left(\sum_{i=1}^n|\{j\in J\mid i\in S_j\}|^{q/p}\right)^{1/q}\\
&=\left(\sum_{i\in\bigcup_{j\in J}S_j}|\{j\in J\mid i\in S_j\}|^{q/p}\right)^{1/q}\\
&=\left(\sum_{i\in I}|\{j\in J\mid i\in S_j\}|^{q/p}\right)^{1/q}&\text{since }\forall j\in J:S_j\subseteq I\\
&\leq \left(|I|^{(q-p)/p}\left(\sum_{i\in I}|\{j\in J\mid i\in S_j\}|\right)^{q/p}\right)^{1/q}&\text{by H\"older}\\
&\leq |I|^{(q-p)/(pq)}\left(\sum_{j\in J}|S_j|\right)^{1/p}\\
&\leq (10s\ln m)^{\delta(p-q)/(pq)}(10s\ln m)^{1/p}\\
&=(10m^{\varepsilon}\ln m)^{(\delta(p-q)+q)/(pq)}.
\end{align*}

Thus, if the conjecture holds for these parameters, then $(p,q)$-Fair Clustering cannot be approximated to within less than the ratio between these bounds (the random optimum divided by the dense optimum), which is at least
\begin{align*}
    20^{-1/q}(10m^{\varepsilon}\ln m)^{(1/q)-(\delta(p-q)+q)/(pq)}&=20^{-1/q}(10m^{\varepsilon}\ln m)^{(1-\delta)(p-q)/(pq)}
    \geq m^{\varepsilon(1-\delta)(p-q)/(pq)}.
\end{align*}
\end{proof}

While not every setting of $0<\varepsilon<\delta<1$ with $\varepsilon\leq\frac12$ in the conjecture is supported by the same evidence (the Sum of Squares result, for instance, gives a much smaller polynomial gap), it is worth noting that any setting yields a lower bound of $m^{\Omega((p-q)/(pq))}$ for $(p,q)$-Fair Clustering, and in particular, if $\varepsilon$ and $\delta$ are arbitrarily close to $\frac12$, as is permissible in Conjecture~\ref{conj:dense-vs-random}, then we get a lower bound of $m^{\left(\frac14-\eta\right)\frac{p-q}{pq}}$ for any constant $\eta>0.$

\end{document}